\documentclass[cleveref,thm-restate]{lipics-v2021}
\hideLIPIcs
\nolinenumbers
\usepackage{graphicx}
\usepackage{xcolor}
\usepackage{xspace}
\usepackage{bbm}
\usepackage{mathtools}
\usepackage[style=english]{csquotes}

\usepackage[ruled, linesnumbered, vlined]{algorithm2e}
\DontPrintSemicolon
\usepackage{multirow}

\usepackage{thm-restate}

\newcommand{\Oh}{\ensuremath{\mathcal{O}}}

\newcommand{\F}{{Fréchet distance}\xspace}

\newcommand{\complSimp}{\ensuremath{\ell}} 
\newcommand{\FreeOld}{\ensuremath{F_{\text{pre}}}}
\newcommand{\FreeNew}{\ensuremath{F_{\text{new}}}}
\newcommand{\updates}{\ensuremath{u}} 
\newcommand{\setA}{\ensuremath{A}} 
\newcommand{\numbA}{\ensuremath{a}}
\newcommand{\eps}{\ensuremath{\varepsilon}}
\newcommand{\dF}{\ensuremath{d_{\text{F}}}}
\newcommand{\ddF}{\ensuremath{d_{\text{dF}}}}
\newcommand{\lipicsmathdash}{\vcenter{\hbox{\labelitemi}}}

\usepackage{soul}

\title{Fréchet Distance in the Imbalanced Case}
\titlerunning{Fréchet Distance in the Imbalanced Case}

\authorrunning{L. Blank}

\Copyright{Lotte Blank} 

\author{Lotte Blank}{University of Bonn, Germany}{lblank@uni-bonn.de}{https://orcid.org/0000-0002-6410-8323}{Funded by the Deutsche Forschungsgemeinschaft (DFG, German Research Foundation) – 459420781 (FOR AlgoForGe)}

\begin{document}

\maketitle

\begin{abstract}
Given two polygonal curves $P$ and $Q$ defined by $n$ and $m$ vertices with $m\leq n$, we show that the discrete Fréchet distance in 1D cannot be approximated within a factor of $2-\eps$ in $\mathcal{O}((nm)^{1-\delta})$ time for any $\varepsilon, \delta>0$ unless OVH fails. 
Using a similar construction, we extend this bound for curves in 2D under the continuous or discrete Fréchet distance and increase the approximation factor to $1+\sqrt{2}-\varepsilon$ (resp. $3-\varepsilon$) if the curves lie in the Euclidean space (resp. in the $L_\infty$-space). This strengthens the lower bound by Buchin, Ophelders, and Speckmann to the case where $m=n^{\alpha}$ for $\alpha\in(0,1)$ and increases the approximation factor of $1.001$ by Bringmann.
For the discrete Fréchet distance in 1D, we provide an approximation algorithm with optimal approximation factor and almost optimal running time.
Further, for curves in any dimension embedded in any $L_p$ space, we present a $(3+\varepsilon)$-approximation algorithm for the continuous and discrete Fréchet distance using $\mathcal{O}((n+m^2)\log n)$ time, which almost matches the approximation factor of the lower bound for the $L_\infty$ metric.
\end{abstract}

\section{Introduction}
The \F is a distance metric for polygonal curves, where each curve is represented as a piecewise-linear interpolation between its vertices. There are different variants of this metric and the two most common variants are the discrete and the continuous Fréchet distance. Alt and Godau~\cite{AG95} were the first to study the computational complexity of the Fréchet distance. Let $n$ and $m$ denote the number of vertices of the two curves, which we refer to as their complexities, with $n \ge m$. They showed that the continuous Fréchet distance can be computed in $\Oh(nm \log n)$ time. Eiter and Mannilla~\cite{EM94} defined the discrete \F and gave an $\Oh(nm)$ time algorithm. Although both results were slightly improved later~\cite{CH25, BBMM2017, AAK2014}, Bringmann~\cite{bringmann2014walking} showed that for $d$-dimensional curves with $d \ge 2$ it is very unlikely that significantly faster algorithms exists: unless the Strong Exponential Time Hypothesis (SETH) fails, no $(1.001)$-approximation algorithm for the continuous or discrete Fréchet distance can run in $\Oh((nm)^{1-\delta})$ time for any $\delta > 0$.
Later, Bringmann and Mulzer~\cite{bringmann2016approximability} extended this result for the discrete \F in 1D and increased the approximation factor up to~$1.4-\eps$ for any $\eps>0$\footnote{In~\cite{bringmann2016approximability}, this result is stated only for $n=m$, but the same construction works for the case $n\neq m$ as well.}.
For the continuous and the discrete \F in 1D, Buchin, Ophelders, and Speckmann~\cite{BOS18} showed a hardness result for approximation factors up to~$3-\eps$, but only in the case where $n=m$. A recent result by Blank and Driemel~\cite{BD24} shows that the continuous \F in 1D can indeed be computed faster when $m=n^{\alpha}$ with $\alpha\in(0,1)$, namely in $\Oh(n\log n+m^2\log^2n)$ time. 
\pagebreak
\subparagraph*{Related Work.}
In recent years, many works have studied approximation algorithms for the \F~\cite{CF21, HO24, HKOS23}. This line of research culminated in the first constant-factor approximation algorithm for the \F that runs in subquadratic time, due to Cheng, Huang, and Zhang~\cite{CHZ25}. 
Their randomized $(7+\eps)$-approximation algorithm for the continuous Fréchet distance runs in $\Oh(nm^{0.99}\log(n/\eps))$ time for any $\eps > 0$ and succeeds with probability at least $1-1/n^6$.

A simple idea to $3$-approximate the \F between curves where $m=n^\alpha$ with $\alpha\in(0,1)$ in subquadratic time is the following. Compute a curve $P'$ of complexity $n'\in o(n)$ and its distance~$\delta$ to $P$ such that every curve $Q$ of complexity~$m$ has \F at least~$\delta$ to $P$. 
Then compute the \F $\delta'$ between $P'$ and $Q$. Using triangle inequality, one can show that the \F between $P$ and $Q$ is at least $\max\{\delta'/2, \delta\}$ and at most $3\cdot \max\{\delta'/2, \delta\}$, which yields a $3$-approximation. The difficulty of this approach is to compute such a curve $P'$. 
Two related variants of this simplification problem have been studied.
The first is computing an optimal \emph{$m$-simplification} of~$P$, that is, a curve~$P^*$ of complexity at most~$m$ minimizing the \F to~$P$ among all curves of complexity at most~$m$. The second is computing an optimal \emph{$\delta$-simplification}, that is, a curve~$P^*$ of minimum complexity among all curves that have \F at most~$\delta$ to~$P$. However, no previous result achieved a subquadratic running time for curves in arbitrary dimension when the approximation factor on the complexity of the simplified curve is constant~\cite{AHMW05, BJWYZ08, BC19, CH23,CHJ25, FF23, KLW18}.

A different approach to reduce the running time of computing the \F is via \F oracles. Such an oracle is a data structure that preprocesses an input curve and subsequently answers \F queries with respect to query curves. Driemel and Har-Peled~\cite{DH13} described a data structure of size $\Oh(n\log n)$ using $\Oh(m^2\log n\log(m\log n))$ query time, which approximates the continuous \F up to a (large) constant factor. For exact queries, \cite{CH24} presented a data structure of size $\Oh(nm)^{\text{poly}(d,m)}$ and query time in $\Oh((md)^{\Oh(1)}\log(nm))$, where the query complexity~$m$ is given at preprocessing time. 
Blank and Driemel~\cite{BD24} showed that for curves in 1D, there is a data structure of size $\Oh(n\log n)$ that computes the continuous Fréchet distance exactly in $\Oh(m^2\log^2 n)$ time.
For the discrete \F, Filtser and Filtser~\cite{FF23} improved~\cite{DPS19} and presented a data structure with size in $\mathcal{O}(1/\eps)^{kd}\log (\eps)^{-1}$ that approximates the discrete \F within a factor of $(1+\eps)$ using expected $\widetilde{\mathcal{O}}(kd)$ query time, where $\widetilde{\Oh}(\cdot)$ hides polylogarithmic factors in $n$. Gudmundsson, Seybold, and Wong~\cite{GSW24} used the lower bound construction by Bringmann~\cite{bringmann2014walking} to show in a similar fashion to \cite{BDNP21, DPS19, R18} the following.
In 2D, there exists no \F oracle with $\textup{poly}(n)$ preprocessing time and $\Oh((nm)^{1-\delta})$ query time and approximation factor~$1.001$ unless SETH fails. 

\begin{table}
\caption{Upper and Lower Bounds for the \F Computation}\label{tab:overviewResults}
    \label{tab:placeholder}
\begin{center}
    \begin{tabular}{|c|clc|clc|}
        \hline
        \multicolumn{7}{|c|}{Upper Bounds}\\
        \hline
          & \multicolumn{3}{|c|}{Discrete} & \multicolumn{3}{|c|}{Continuous}\\
          \hline
          \multirow{2}{1.5em}{1D}&${\Oh}(mn)$ &exact &\cite{EM94, AAK2014}& $\widetilde{\Oh}(mn)$ &exact &\cite{AG95,BBMM2017, CH25}\\
         &$\widetilde{\Oh}(n+m^2)$ &$2$-approx. &Cor.\ref{cor:1DUpper}& $\widetilde{\Oh}(n+m^2)$ &exact &\cite{BD24}\\
         \hline
         \multirow{3}{1.5em}{dD}&${\Oh}(mn)$ &exact&\cite{EM94, AAK2014}& $\widetilde{\Oh}(mn)$&exact &\cite{AG95,BBMM2017,CH25}\\
            &$\widetilde{\Oh}(n+m^2)$& $(3+\eps)$-approx. &Thm.\ref{thm:dDUpper}&$\widetilde{\Oh}(n+m^2)$& $(3+\eps)$-approx. &Thm.\ref{thm:dDUpper}\\
            &$\widetilde{\Oh}(nm^{0.99})$ &$(7+\eps)$-approx. &\cite{CHZ25}&$\widetilde{\Oh}(nm^{0.99})$ &$(7+\eps)$-approx. &\cite{CHZ25}\\
         \hline\hline
         \multicolumn{7}{|c|}{Lower Bounds}\\
        \hline
          & \multicolumn{3}{|c|}{Discrete} & \multicolumn{3}{|c|}{Continuous}\\
          \hline
         \multirow{3}{1.5em}{1D}&$\nexists~{\Oh}(m^{2-\delta})$ &$(3-\eps)$-approx.  &\cite{BOS18}& $\nexists~{\Oh}(m^{2-\delta})$ &$(3-\eps)$-approx. &\cite{BOS18}\\
         &$\nexists~\widetilde{\Oh}((mn)^{1-\delta})$ &$(1.4-\eps)$-approx. &\cite{bringmann2016approximability}& &&\\
         &$\nexists~\widetilde{\Oh}((mn)^{1-\delta})$ &$(2-\eps)$-approx. &Thm.\ref{thm:1D}& &&\\
         \hline
         \multirow{3}{1.5em}{2D}&  
         \multicolumn{6}{|l|}{\hspace{2.5cm}${\nexists~\Oh}((mn)^{1-\delta})$\hspace{0.8cm} $(1.001)$-approx. \hspace{2cm}\cite{bringmann2014walking}}\\
         &\multicolumn{6}{|l|}{\hspace{2.5cm}${\nexists~\Oh}((mn)^{1-\delta})$\hspace{0.8cm} $(1+\sqrt{2}-\eps)$-approx. in $L_2$ \hspace{0.3cm}Thm.\ref{thm:2D}}\\
        &\multicolumn{6}{|l|}{\hspace{2.5cm}${\nexists~\Oh}((mn)^{1-\delta})$\hspace{0.8cm} $(3-\eps)$-approx. in $L_\infty$ \hspace{1cm}Thm.\ref{thm:2D}}\\
        \hline
    \end{tabular}
    
\end{center}
\end{table}

\subparagraph*{Our Results.}
 Our first set of results shows lower bounds for algorithms that approximate the \F conditioned on the Orthogonal Vectors Hypothesis (OVH).
 We show that the discrete \F between 1-dimensional curves cannot be $(2-\eps)$-approximated in $\Oh((nm)^{1-\delta})$ time for any $\eps, \delta>0$ unless OVH fails. This increases the approximation factor of $1.4-\eps$ by Bringmann and Mulzer~\cite{bringmann2016approximability}. Using a similar construction in 2D, we show that the continuous and discrete \F cannot be approximated within a factor less than $1+\sqrt{2}$ (resp. $3$) when the curves lie in the Euclidean space (resp. in the $L_\infty$-space) in $\Oh((nm)^{1-\delta})$. This leads to lower bounds for \F oracles as well and answers an open question by~\cite{GSW24}.

To obtain this first set of results, we reduce from the \emph{Orthogonal Vector Problem:} Given two sets of vectors $U, V\subset \{0,1\}^d$ of size $m$ and $n$, do there exist two vectors $u\in U$ and $v\in V$ such that $\langle u, v\rangle=0$? 
Then, the results follow by OVH, which is implied by SETH~\cite{Wil05, BK18}.

\begin{definition}[Orthogonal Vectors Hypothesis (OVH)]\label{def:OVH}
    For all $\delta>0$ there exists a $c>0$ such that there is no algorithm solving OV instances $U,V\subset \{0,1\}^d$ of size $m$ and $n$ with $m\leq n$ and $d=c\log n$ in time $\mathcal{O}((nm)^{1-\delta})$.
\end{definition}
We construct a curve~$P$ of complexity $\Oh(n+m)$ depending on $V$ and the size of $U$ and a curve~$Q$ of complexity $\Oh(m)$ depending only on $U$. Then, we show that the \F between~$P$ and $Q$ is at most $1$ if and only if the OV instance has a solution. In order to get a lower bound also for the case where $n\gg m$, we construct $P$ such that the main part of~$P$ lies in a ball of radius~$1$. Hence, this main part can be matched to a single point of~$Q$. Further, to obtain an inapproximability result of $2-\eps$ in 1D, the points of~$P$ and~$Q$ have integer coordinates. Hence, there are only three different values that the vertices of the main part of $P$ can have. 
This significantly restricts our possibilities to define~$P$ and is the main difference to the construction by Bringmann and Mulzer~\cite{bringmann2016approximability}. They also note in their paper that obtaining a lower bound for an approximation factor of $1.4$ or higher requires a significantly different construction compared to theirs.

On the positive side, we present a data structure that preprocesses a 1-dimensional curve of complexity $n$ in $\mathcal{O}(n\log n)$ time using $\mathcal{O}(m)$ space that $2$-approximates the discrete \F to query curves of complexity $m$ in $\Oh(m^2\log m)$ time. Hence, there exists a $2$-approximation algorithm for the discrete \F in $\Oh(n\log n+m^2\log m)$ time matching our approximation factor lower bound and essentially the time lower bound by~\cite{BOS18}. Secondly, we match our approximation factor lower bound in the $L_\infty$-space by giving a simple $(3+\eps)$-approximation algorithm for the discrete and continuous \F using $\Oh((n+m^2)\log n)$ time for curves in any dimension under any $L_p$-norm. 
\Cref{tab:overviewResults} summarizes the results of this paper and related work.

\section{Preliminaries}
For any two points $p_1, p_2\in \mathbb{R}^d$, $\overline{p_1 p_2}$ denotes the directed line segment connecting $p_1$ with~$p_2$. 
A polygonal curve $P$ of complexity $n$ is formed by ordered line segments $\overline{P(i) P(i+1)}$ of points $P(1), P(2), \dotso, P(n)$, where $P(i)\in \mathbb{R}^d$. 
It can be viewed as a function $P:[1, n]\rightarrow \mathbb{R}^d$, where $P(i+\alpha)=(1-\alpha)P(i)+\alpha P(i+1)$ for $i\in \{1, \dotso, n\}$ and $\alpha\in[0,1]$. This curve is also denoted with $\langle P(1), \dotso, P(n)\rangle$. We call the points $P(i)$ vertices and the line segments $\overline{P(i) P(i+1)}$ edges of $P$. Further, $P[s, t]$ is the subcurve of~$P$ obtained from restricting the domain to the interval~$[s, t]$. Two curves $P=\langle P(1), \dots, P(n)\rangle$ and $Q=\langle Q(1), \dots, Q(m)\rangle$ can be composed into the curve $P\circ Q=\langle P(1), \dots, P(n), Q(1), \dots, Q(m)\rangle$.
For a number~$k\in \mathbb{N}_0$, define $P^k$ be the composition of $k$ copies of $P$. 

Let $P: [1,n]\rightarrow \mathbb{R}^d$ and ${Q:[1,m]\rightarrow \mathbb{R}^d}$ be two polygonal curves. 
 A \emph{continuous matching}~$M$ between~$P$ and~$Q$ is defined by two functions $h_P\in \mathcal{F}_n$ and $h_Q\in \mathcal{F}_m$ and consists of the tuples $(h_P(a), h_Q(a))$ for $a\in [0,1]$, where $\mathcal{F}_n$ is the set of all continuous, non-decreasing functions $h: [0,1]\rightarrow [1,n]$ with ${h(0)=1}$ and $h(1)=n$, respectively $\mathcal{F}_m$ for~$m$. 
 Then, the \emph{(continuous) \F} between~$P$ and~$Q$ is defined as
\[\dF(P, Q)=\min_{\text{continuous matching }M}\ \max_{(x,y)\in M} \| P(x)-Q(y)\|.\]
For a fixed~$\delta$, the \emph{free space} $F_\delta(P, Q)$ is a subset of $[1, n]\times [1, m]$ such that $(x, y)\in F_\delta(P, Q)$ if and only if $\|P(x)-Q(y)\|\leq \delta$. Points in $([1, n]\times [1, m])\setminus F_\delta(P, Q)$ are contained in the non-free space. Hence, the continuous \F between $P$ and $Q$ is at most~$\delta$ if a continuous matching $M\subset F_\delta(P, Q)$ between $P$ and $Q$ exists.

A \emph{discrete matching}~$M$ between $P$ and $Q$ is an ordered sequence that starts in $(1,1)$, ends in $(n,m)$ and for any consecutive tuples $(i, j)$, $(i', j')$ in $M$ it holds that $(i', j')\in \{(i+1, j), (i, j+1), (i+1, j+1)\}$. Then, the \emph{discrete Fréchet distance} between $P$ and $Q$ is 
\[\ddF(P, Q)=\min_{\text{discrete matching } M}\ \max_{(i,j)\in M} \|P(i)-Q(j)\|.\]
The \emph{free space matrix} $M_{\delta}$ is an $n\times n$ matrix, where $M_{\delta}[i,j]=1$ if $|P(i)-Q(j)|\leq \delta$ and otherwise $M_{\delta}[i,j]=0$. We say that $(i,j)$ is \emph{reachable} if $\ddF(P[1, i], Q[1, j])\leq \delta$.

\begin{observation}\label{o:concatanation}
    Let $P_1, P_2, Q_1, Q_2$ be polygonal curves.
    \begin{itemize}
        \item If $\dF(P_1, Q_1)\leq \delta$ and ${\dF(P_2, Q_2)\leq\delta}$, then ${\dF(P_1\circ P_2,Q_1\circ Q_2)\leq \delta}$. 
        \item If $\ddF(P_1, Q_1)\leq \delta$ and $\ddF(P_2, Q_2)\leq\delta$, then ${\ddF(P_1\circ P_2,Q_1\circ Q_2)\leq \delta}$.
    \end{itemize}
\end{observation}
\begin{observation}\label{obs:discreteContinuous}
    Let $P$ and $Q$ be polygonal curves, then $\dF(P,Q)\leq \ddF(P, Q)$.
\end{observation}

\section{Approximation Lower Bounds for the Imbalanced Case}
In this section, we show a lower bound for the discrete \F in 1D and for the discrete and continuous \F in 2D depending on OVH.
The construction for both bounds is very similar and we start with giving a schematic view of the reduction.

Let $U,V\subset \{0,1\}^d$ be two sets of size~$m$ and~$n$, where $m\leq n$ and denote with $v_{i k}$ (resp.~$u_{ik}$) the $k$-th bit of the $i$-th vector of $V$ (resp. $U$).
We construct two 1-dimensional (resp. 2-dimensional) curves~$P$ and~$Q$ of complexity $\mathcal{O}(nd)$ and $\mathcal{O}(md)$, where~$P$ depends on~$V$ and~$Q$ on~$U$. We show that their discrete (and continuous) Fréchet distance is at most~$1$ if and only if two vectors that are orthogonal exist in~$V$ and~$U$. 
Define
\[P=P_s\circ(P^*)^{m-1}\circ P_{null}\circ\bigcirc_{i=1}^{n}\left( (P_0)^{d+1}\circ P_1\circ \bigcirc_{\ell=1}^{d}P_{v_{i\ell}}\circ P_1\right)\circ P_{null}\circ(P^*)^{m-1}\circ P_s,\]
\[Q=(Q^*)^{m-1} \circ\bigcirc_{j=1}^{m}\left(Q_{null}\circ Q_c\circ (Q_1)^{d+1}\circ Q_0\circ \bigcirc_{\ell=1}^{d}Q_{u_{j\ell}}\circ Q_0\circ Q_c\right)\circ Q_{null}\circ (Q^*)^{m-1},\]
\addtocounter{linenumber}{-3}
where the subcurves are defined later.
Since the complexity of~$P$ can be arbitrarily larger than the complexity of $Q$, we construct the curve~$P$ such that the main part of~$P$, i.e., in between the two subcurves $P_{null}$, is contained in a ball of radius~$1$. Then, we can match this whole part to a single point on~$Q$. This point is $Q_c$ in our construction.
At a high level idea, we want to enforce that in a matching that realizes the \F, the two subcurves~$P_{null}$ get roughly matched to different subcurves~$Q_{null}$. Then, the subcurve in between the two points matched to~$P_{null}$ contains
$(Q_1)^{d+1}$ as a subcurve. 
Now, we enforce that $(Q_1)^{d+1}$ gets roughly matched to $(P_0)^{d+1}$. Hence, afterwards $\bigcirc_{k=1}^{d}P_{v_{ik}}$ and $\bigcirc_{\ell=1}^{d}Q_{u_{j\ell}}$ have to be matched together for some $i$ and $j$.
We will show that this is only possible if $\langle v_i, u_j\rangle=0$. This matching is schematically visualized in \Cref{fig:LowerBoundSchematicView}.
\begin{figure}
    \centering
    \includegraphics[page=2, width=\textwidth]{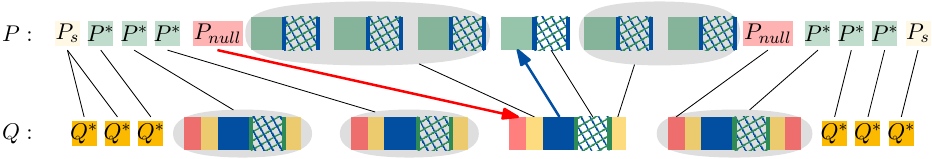}
    \caption{A schematic view of the construction of the curves $P$ and $Q$. Here, $Q_{null}$ is drawn in red, $P_0$ and $Q_0$ in green,  $P_1$ and  $Q_1$ in blue, and  $Q_c$ in orange.}
    \label{fig:LowerBoundSchematicView}
\end{figure}
We define
\begin{align*}
    &P_s=\langle(-1,0)\rangle,
    && P_{null}=\langle (-5,0) \rangle,
    &&P^*=\langle (-3,0)\rangle\circ\langle P_0\rangle^{2d+3}\circ \langle (-3,0)\rangle, \\
    & Q_c=\langle(0,0)\rangle,
    && Q_{null}=\langle (-4,0)\rangle, 
    && Q^*=\langle (-2,0),(0,0),(-2,0)\rangle,
\end{align*}
 for 2D and the definition carries over directly to 1D.
The curves $P_0$, $P_1$, $Q_0$, and $Q_1$ are visualized in \Cref{fig:construction_P0P1Q0Q1}. For the 1D setting, let
\begin{align*}
    & P_0=\langle -1, 1\rangle,
    &&P_1=\langle -1, 0\rangle,& Q_0=\langle -2,1\rangle,
    && Q_1=\langle -2, 2\rangle.
\end{align*}

For the 2D constructions, define eight points with distance $1$ to $(0,0)$ where $a_1=(-1, 0)$ and the direction of $a_{i+1}$ is the direction of $a_i$ rotated by $\pi/4$, i.e., in~$L_2$
\begin{align*}
    &a_1=(-1, 0), &&a_2=(-1/\sqrt{2}, -1/\sqrt{2}), &&a_3=(0, -1),&&a_4=(1/\sqrt{2}, -1/\sqrt{2}),
    \\&a_5=(1, 0),&&a_6=(1/\sqrt{2}, 1/\sqrt{2}), &&a_7=(0, 1),&&a_8=(-1/\sqrt{2}, 1/\sqrt{2}),
\end{align*}
and in~$L_\infty$
\begin{align*}
    &a_1=(-1, 0), &&a_2=(-1, -1), &&a_3=(0, -1),&&a_4=(1, -1),
    \\&a_5=(1, 0),&&a_6=(1, 1), &&a_7=(0, 1),&&a_8=(-1, 1).
\end{align*}
Define $c=\sqrt{2}$ and $\tilde{a}=a_8+a_7$ in $L_2$ and $c=2$ and $\tilde{a}=2a_8$ in $L_\infty$. Note that~$c$ is defined such that $ca_1$ is the point with distance 1 to~$a_8$ and~$a_2$ that maximizes the distance to~$(0,0)$ and $\tilde{a}$ such that $\|\tilde{a}-a_2\|=1+c$ and $\|\tilde{a}-a_8\|=1$.
Then we set
\begin{align*}
    & P_0=\langle a_1, a_2,\dots a_7, a_8, a_7\dots, a_1\rangle,
    &&P_1=\langle a_1, a_2, \dots, a_8, a_1, a_2, a_1 a_8,  \dots, a_1\rangle,\\
    & Q_0=\langle 2a_1, 2a_2, \dots, 2a_8, ca_1, 2a_8, \dots, 2a_1\rangle,
    &&Q_1=\langle 2a_1, 2a_2, \dots, 2a_7, \tilde{a}, 2a_7, \dots, 2a_1\rangle.
\end{align*}
We show that this results in an inapproximability result of $(1+c-\eps)$ for time in $\Oh((nm)^{1-\delta})$.

\begin{figure}
    \centering
    \includegraphics[page=10, width=0.24\linewidth]{Figures/Lower_Bound_Curves.pdf}\ 
    \includegraphics[page=11, width=0.24\linewidth]{Figures/Lower_Bound_Curves.pdf}\ 
    \includegraphics[page=12, width=0.24\linewidth]{Figures/Lower_Bound_Curves.pdf}\ 
    \includegraphics[page=13, width=0.24\linewidth]{Figures/Lower_Bound_Curves.pdf}\\
    \ \\
    \centering
    \includegraphics[page=1, width=0.24\textwidth]{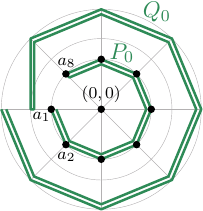}\ 
    \includegraphics[page=2, width=0.24\textwidth]{Figures/Lower_Bound_Curves.pdf}\ 
    \includegraphics[page=3, width=0.24\textwidth]{Figures/Lower_Bound_Curves.pdf}\ 
    \includegraphics[page=4, width=0.24\textwidth]{Figures/Lower_Bound_Curves.pdf}\\
    \ \\
    \includegraphics[page=5, width=0.24\textwidth]{Figures/Lower_Bound_Curves.pdf}
    \includegraphics[page=6, width=0.24\textwidth]{Figures/Lower_Bound_Curves.pdf}\ 
    \includegraphics[page=7, width=0.24\textwidth]{Figures/Lower_Bound_Curves.pdf}\ 
    \includegraphics[page=8, width=0.24\textwidth]{Figures/Lower_Bound_Curves.pdf}\ 
    \caption{The curves $P_0$, $P_1$, $Q_0$, and $Q_1$. The upper curves are defined in 1D, the middle curves 
    for the $L_2$-metric in 2D and the lower curves for the $L_\infty$-metric in 2D.}
    \label{fig:construction_P0P1Q0Q1}
\end{figure}

\begin{lemma}\label{lem:LBdirection1}
    If there exist indices $i$ and $j$ such that $\langle v_i, u_j\rangle=0$, then in every construction for $P$ and $Q$ above it holds that $\ddF(P, Q)\leq 1$.
\end{lemma}
\begin{proof}
    For every construction of $P$ and $Q$, it holds that $\ddF(P^*, Q^*)\leq 1$, $\ddF(P_s, Q^*)\leq 1$. 
    Further, for $x, y\in \{0,1\}$ with $\langle x, y\rangle=0$, it holds that $\ddF(P_x, Q_y)\leq1$.
    Then, we can match the following with \F at most 1:
    \begingroup
    \setlength{\mathindent}{0pt}
    \begin{align*}
        \lipicsmathdash~&P_s\circ (P^*)^{m-j}&\text{ and }&&
        &(Q^*)^{m-1},
        \\
        \lipicsmathdash~&(P^*)^{j-1}\circ P_{null}&\text{ and }&&
        &\bigcirc_{k=1}^{j-1}(Q_{null}\circ Q_c\circ (Q_1)^{d+1}\circ Q_0\circ \\
        &&&&&~~~~~~~~\bigcirc_{\ell=1}^{d}Q_{u_{k\ell}}\circ Q_0\circ Q_c ){\circ Q_{null}}, ~~~~~
        \\
        \lipicsmathdash~&\bigcirc_{k=1}^{i-1}\left((P_0)^{d+1}\circ P_1\circ \bigcirc_{\ell=1}^{d}P_{v_{k\ell}}\circ P_1\right)&\text{ and }&&
        & Q_c, \\
        \lipicsmathdash~&(P_0)^{d+1}\circ P_1\circ \bigcirc_{\ell=1}^{d}P_{v_{i\ell}}\circ P_1&\text{ and }&&
        &(Q_1)^{d+1}\circ Q_0\circ \bigcirc_{\ell=1}^{d}Q_{u_{j\ell}}\circ Q_0,
        \\
        \lipicsmathdash~&\bigcirc_{k=i+1}^{n}\left((P_0)^{d+1}\circ P_1\circ \bigcirc_{\ell=1}^{d}P_{v_{k\ell}}\circ P_1\right)&\text{ and }&&
        &Q_c,
        \\
        \lipicsmathdash~&P_{null}\circ(P^*)^{m-j}&\text{ and }&&
        &\bigcirc_{k=j+1}^{m}(Q_{null}\circ Q_c\circ (Q_1)^{d+1}\circ Q_0\circ \\
        &&&&&~~~~~~~~\bigcirc_{\ell=1}^{d}Q_{u_{k\ell}}\circ Q_0\circ Q_c){\circ Q_{null}},
        \\
        \lipicsmathdash~&(P^*)^{j-1}\circ P_s&\text{ and }&&
        &(Q^*)^{m-1}.
    \end{align*}
    \endgroup
    By~\Cref{o:concatanation}, it follows that $\ddF(P, Q)\leq 1$.
\end{proof}

To show the other direction, we first prove some lemmas.
\begin{restatable}{lemma}{lemDdiscreteLB}
\label{lem:1DdiscreteLB}
    Consider the 1D setting and let $Q'=\bigcirc_{\ell=1}^{d'}Q_{\sigma_\ell}$, where $\sigma_{\ell}\in\{0,1\}$ for every~$\ell$. Further, let $s\leq t$ be such that $P[s, t]$ is a subcurve in between the two subcurves~$P_{null}$ and $\ddF(Q', P[s,t])<2$. Then
    $P[s,t]=\bigcirc_{\ell=1}^{d'}P_{\pi_\ell}$ where $\pi_{\ell}=0$ if $\sigma_\ell=1$ and $\pi_{\ell}\in\{0,1\}$ if $\sigma_\ell=0$ for every~$\ell$. 
\end{restatable}
\begin{proof}
    It holds that every second vertex of $Q'$ is $-2$ and the vertices of value $-1$ on $P[s, t]$ are the only vertices on $P[s, t]$ within distance less than $2$ to those vertices of $Q'$. Further, it holds that the vertices with value $-1$ of $P[s, t]$ have distance less than two only to the vertices of value $-2$ of $Q'$. Hence,
    since the vertices of value $-1$ on $P[s, t]$ are exactly every second vertex on $P[s, t]$, we must traverse $P[s, t]$ and $Q'$ synchronously. Therefore, if $\sigma_{\ell}=0$ then $\pi_{\ell}\in \{0, 1\}$ and if $\sigma_{\ell}=1$ then $\pi_\ell=0$ by construction.
\end{proof}

In 2D, we also want to show a bound for the continuous \F. Then the vertices of $Q$ can also be matched to points on edges of $P$. We use the notation of the following subcurve of $P_0$ (resp. $P_1$): $P_0'=\langle a_4, a_5,\dots, a_8, \dots, a_4\rangle$ (resp. $P_1'=\langle a_4, \dots, a_8, a_1, a_2, a_1, a_8, \dots, a_4\rangle$).

\begin{lemma}\label{lem:2DcontinuousLB_Q01}
    Consider the 2D setting.
    Then, for any subcurve $P'$ of $P$ in between the two subcurves~$P_{null}$ with $\dF(P', Q_y)<1+c$ for $y\in \{0,1\}$, it holds that $P'_x$ is a subcurve of $P'$ and~$P'$ is a subcurve of $\langle a_4, a_3, a_2\rangle\circ P_x\circ \langle a_2, a_3, a_4\rangle$, where $x=0$ if $y=1$ and $x\in \{0, 1\}$ if $y=0$.
\end{lemma}

\begin{figure}
    \centering
    \includegraphics[page=9, width=0.95\textwidth]{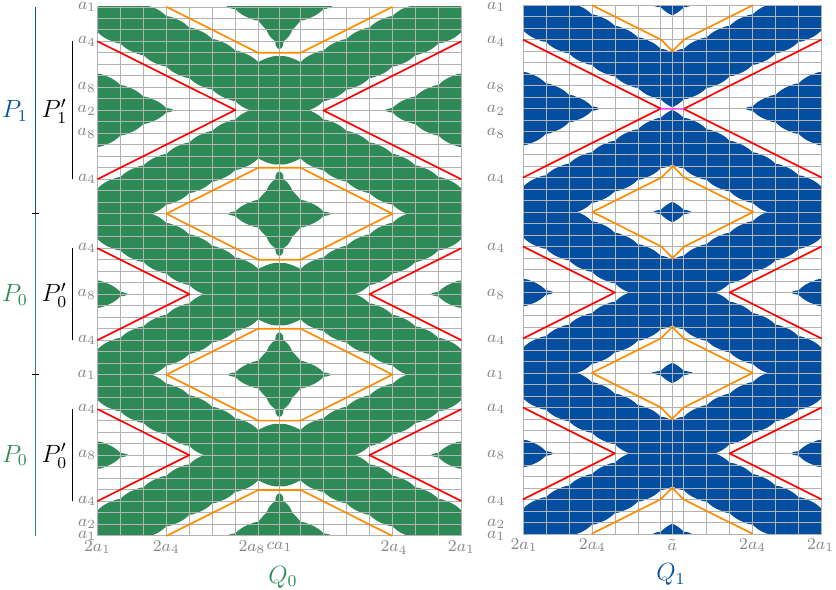}
    \caption{In white: the non-free space of the curve $Q_0$ and $P_0\circ P_0\circ P_1$ on the left and of the curve $Q_1$ and $P_0\circ P_0\circ P_1$ on the right. The red, orange and pink paths visualize the proof of \Cref{lem:2DcontinuousLB_Q01}.}
    \label{fig:freespace}
\end{figure}
\begin{proof}
    We show this lemma by constructing paths in the non-free space for $\delta=1+c$ (see \Cref{fig:freespace}).
    First observe that $\|2a_i-a_{j}\|\geq 1+c$ for any $i, j\in\{1, \dots, 8\}$ with ${i-j\geq 3 \mod 8}$. 

    Assume that the first (resp. last) vertex of $Q_y$ is matched to a point on a subcurve $P'_x$ with $x\in\{0,1\}$. Then, there must exist a path completely contained in the free space ending at the last (resp. first) vertex of $Q_y$. This cannot be the case since there exists a continuous matching between $P'_x$ and $\langle 2a_1, 2a_2\dots, 2a_5, \dots, 2a_1\rangle$ that is completely contained in the non-free space (see the red paths in \Cref{fig:freespace}).
    Therefore, the first as well as the last vertex of $Q_y$ has to be matched to a point on a subcurve~$P''$ of~$P$ such that $P''=\langle a_4, a_3, a_2, a_1, a_2, a_3, a_4\rangle$.

    Next, we show that $P'$ contains exactly one $P'_x$ as a subcurve. Since $\|ca_1-a_5\|\geq 1+c$ (resp. $\|\tilde{a}-a_5\|\geq 1+c$), there exists a path that is completely contained in the non-free space between $\langle a_5, a_4, \dots a_2, a_1, a_2, \dots, a_4, a_5\rangle$ and $\langle 2a_8, ca_1, 2a_8, 2a_7, \dots, 2a_4, 2a_5, \dots, 2a_8, ca_1, 2a_8\rangle$(resp. \linebreak $\langle \tilde{a}, 2a_7, \dots, 2a_4, 2a_5, \dots, \tilde{a}\rangle$) (see the orange paths in \Cref{fig:freespace}). Hence by construction of~$P$, it follows that~$P'_x$ is a subcurve of~$P'$ and~$P'$ is a subcurve of $\langle a_4,a_3, a_2\rangle\circ P_x\circ \langle a_2, a_3, a_4\rangle$ with $x\in \{0,1\}$.
    It remains to show that $x\neq 1$, if $y=1$. This follows since the distance between $a_2$ and $2a_7$, $\tilde{a}$,~$2a_7$ is at least~$1+c$ (see the pink path in \Cref{fig:freespace}).
\end{proof}
The next lemma follows directly by the construction of the curves and \Cref{lem:2DcontinuousLB_Q01}.
\begin{lemma}\label{lem:2DcontinuousLB}
    Consider the 2D setting and let $Q'=\bigcirc_{\ell=1}^{d'}Q_{\sigma_\ell}$, where $\sigma_{\ell}\in\{0, 1\}$ for every~$\ell$. Further, let $s\leq t$ be such that $P[s, t]$ is a subcurve in between the two subcurves~$P_{null}$ and $\dF(Q', P[s,t])<1+c$. Then, $P_{\pi_1}' \circ \bigcirc_{\ell=2}^{d'-1}P_{\pi_\ell}\circ P_{\pi_{d'}}'$ is a subcurve of
    $P[s,t]$ and $P[s, t]$ is a subcurve of $\langle a_4,a_3, a_2\rangle \circ \bigcirc_{\ell=1}^{d'}P_{\pi_\ell}\circ \langle a_2, a_3, a_4\rangle$ where $\pi_{\ell}=0$ if $\sigma_\ell=1$ and $\pi_{\ell}\in\{0, 1\}$ if $\sigma_\ell=0$ for every~$\ell$. 
\end{lemma}

\begin{lemma}\label{lem:LBdirection2}
    If $\ddF(P,Q)< 2$ (resp. $\dF(P, Q)<1+c$) in the 1D (resp. 2D) construction, then there exist two indices $i$ and $j$ such that $\langle v_i, u_j\rangle=0$.
\end{lemma}
\begin{proof}
    Let~$M$ be a discrete (resp. continuous) matching that realizes the discrete (resp. continuous) Fréchet distance between $P$ and $Q$. Assume that $\dF(P,Q)< 2$ (resp. $\ddF(P, Q)<1+c$), then in the matching $M$, the two points of value $P_{null}$ of $P$ must be matched to two distinct points $q_1, q_2$ on~$Q$ both on an edge that contains the value $Q_{null}$, since they are the only points within distance at most $1+c\leq 3$ to~$P_{null}$. 
    Let $P_0'=P_0$ and $P_1'=P_1$ in the 1D setting. 
    Define $j$ such that $q_1$ is contained on an incident edge to the $j$-th $Q_{null}$ on $Q$. 
    Then, the first $(Q_1)^{d+1}$ after the $j$-th $Q_{null}$ on $Q$ must be matched to a subcurve $P'$ of $P$ such that $P_1'\circ (P_0)^{d+1}\circ P_1'$ is a subcurve of $P'$ by \Cref{lem:1DdiscreteLB} and \Cref{lem:2DcontinuousLB}. The only time we have $(P_0)^{d+1}$ in $P$ is before $P_1\circ\bigcirc_{\ell=1}^{d}P_{v_{i\ell}}$ for an $i$. Hence by \Cref{lem:1DdiscreteLB} and \Cref{lem:2DcontinuousLB}, there exists an index~$i$ such that in the matching~$M$ the subcurve
    $P'_{v_{i,1}}\circ \bigcirc_{\ell=2}^{d-1}P_{v_{i,\ell}}\circ P'_{v_{i,d}}$ is matched to a subcurve of $\bigcirc_{\ell=1}^{d}Q_{u_{j,\ell}}$ and $v_{i, \ell}=0$ if ${u_{j,\ell}}=1$. 
    Hence, it holds that $\langle v_i, u_j\rangle=0$.
\end{proof}
By \Cref{lem:LBdirection1} and \Cref{lem:LBdirection2} together with \Cref{obs:discreteContinuous}, the next two theorems follow.
\begin{theorem}\label{thm:1D}
    For any $\delta,\eps>0$, there does not exist a $(2-\eps)$-approximation algorithm for the discrete \F between 1-dimensional curves of complexity $n$ and $m$ using $\mathcal{O}((nm)^{1-\delta})$ time unless OVH fails. 
\end{theorem}
\begin{theorem}\label{thm:2D}
    For any $\delta,\eps>0$, there does not exist a $(k-\eps)$-approximation algorithm for the continuous or discrete \F between 2-dimensional curves of complexity $n$ and $m$ using $\mathcal{O}((nm)^{1-\delta})$ time unless OVH fails, where $k=3$ (resp. $k=1+\sqrt{2}\approx2.41$) if the curves lie in the $L_\infty$-space (resp. in the Euclidean space).
\end{theorem}
Our construction gives also a bound on \F oracles.
\begin{restatable}{theorem}{thmFreOracles}
    Assume OVH holds true. For any $\delta>0$, there exists a $c>0$ such that for any $n$ and $m>c\log n$ there is no data structure storing a curve of complexity~$n$ with $\textup{poly}(n)$ preprocessing time, that when given a query curve of at most complexity~$m$ can answer continuous or discrete \F queries in $\mathcal{O}((nm)^{1-\delta})$ query time within the same approximation factor as in \Cref{thm:1D} and \Cref{thm:2D}.
\end{restatable}
\begin{proof}
    Suppose for the sake of contradiction that there exists a data structure with $\Oh(n^\alpha)$ preprocessing time and $\mathcal{O}((nm)^{1-\delta})$ query time. We assume without loss of generality that $\alpha\geq 2$. 
    Let $U$, $V$ be a OV instance with $|U|=M$ and $|V|=n$ of dimension $c \log n$ with $M\geq n^{\alpha-1.5}$ and $M\geq m$, where $c$ is as in \Cref{def:OVH}. Then, we preprocess the curve~$P$ defined by~$V$ as in the lower bound construction above. The set $U$ gets subdivided into $k=\lceil (60c\log n) M/m\rceil$
    disjoint sets $U_1, \dots, U_k$ such that $U=\bigcup_{i=1}^{k}U_i$, $|U_i|=m/(60c\log n)$ for $i<k$ and $|U_k|\leq m/(60c\log n)$. For every set~$U_i$, we construct a query curve~$Q_i$ with the lower bound construction of~$Q$ above. Then, $Q_i$ has complexity at most $m$. For every query curve~$Q_i$, we compute its distance to $P$ in $\mathcal{O}((nm)^{1-\delta})$ time. By the proof of \Cref{thm:1D} and \Cref{thm:2D}, it holds that there exists a curve $Q_i$ within distance $1$ to~$P$ if and only if OV has a solution. Hence, the OV instance is solved in $\Oh(n^\alpha+k(nm)^{1-\delta})$ time, which contradicts OVH.
\end{proof}

\section{2-Approximation for the discrete \F in 1D}
Given a value $m$ and a 1-dimensional curve $P$ of complexity $n\geq m$. We want to preprocess~$P$ in near-linear time and store information about $P$ using $\mathcal{O}(m)$ space such that for any 1-dimensional query curve $Q$ of complexity at most $m$ we can approximate the discrete \F between $P$ and $Q$ in $\mathcal{O}(m^2\log m)$ time.

\subparagraph*{Preprocessing.}
Define $\delta_m=\min\{\ddF(P, Q)\mid Q \text{ has complexity }m\}$. Then, we want to compute a curve~$P'$ such that $\ddF(P, P')\leq \delta_m/2$. The complexity of $P'$ might be~$n$, but it should be possible to store a compact description of $P'$ in $\mathcal{O}(m)$ space.
To compute~$P'$, we first compute the optimal $m$-simplification $P^*$ of $P$ for the discrete Fréchet distance. The curve~$P^*$ is computed via searching over the possible values~$\delta$ for $\delta_m$ and then computing the optimal~$\delta$-simplification. We compute the optimal $\delta$-simplification of $P$ using \Cref{Alg:deltaSimpl}. It begins by finding the longest prefix $P[1, i_1]$ that is contained in a $2\delta$-interval, whose center becomes the first vertex of the $\delta$-simplification. Then, it continues by finding the largest index $i_2$ such that $P[i_1+1, i_2]$ is contained in a $2\delta$-interval and again the center of this interval is added to the $\delta$-simplification. This process is repeated until the final vertex of~$P$ is reached.

\begin{algorithm}
    \caption{Optimal $\delta$-simplification of $P$\label{Alg:deltaSimpl}}
    $\ell\gets 1$, $a\gets P(1)$, $b\gets P(1)$\;
    \For{$i=1$ to $n$}{
        \If{$P(i)\leq a+2\delta$ and $b-2\delta\leq P(i)$\label{li:simplInterval1}}{
            $a\gets\min\{a, P(i)\}$, $b\gets\max\{b, P(i)\}$\label{li:simplInterval2}\;
        }
        \Else{
            $p_\ell\gets\frac{a+b}{2}$, 
            $a\gets P(i)$, $b\gets P(i)$, $i_{\ell}\gets i-1$, $\ell\gets \ell+1$\;
        }
    }
    Return $P_\delta=\langle p_1, \dots, p_\ell\rangle$, $i_0=0, i_1, \dots, i_\ell=n$\;
\end{algorithm}

\begin{restatable}{lemma}{lemdeltaSimpl}
    \label{lem:deltaSimpl}
    \Cref{Alg:deltaSimpl} computes an optimal $\delta$-simplification $P_\delta$ of $P$ using~$\mathcal{O}(n)$~time.
\end{restatable}
\begin{proof}
    First, observe that $\ddF(P, P_\delta)\leq \delta$, since $((1,1), (2,1) \dots, (i_1, 1), (i_1+1, 2), ({i_1+2}, 2),$ $ \dots, (i_2, 2), ({i_2+1}, 3), \dots, (i_\ell, \ell))$ is a discrete matching realizing this distance. Assume that there exists a curve $P'$ of complexity $k<\ell$ such that $\ddF(P, P')\leq \delta$. Then, there exists a discrete matching $M$ between $P$ and $P'$ realizing the discrete \F and a tuple $(x, y)$ in $M$ with $x>i_y$. Let $(x, y)$ be the first such tuple. Then, by line~\ref{li:simplInterval1}--\ref{li:simplInterval2} of \Cref{Alg:deltaSimpl} there exist $P(x_1)\leq P(x_2)$ such that $(x_1, y)$ and $(x_2, y)$ in $M$ and $P(x_1)$, $P(x_2)$, and $P(x)$ are not contained in a $2\delta$ interval. Hence, $\ddF(P, P')>\delta$.
\end{proof}

If the complexity of the $\delta$-simplification is greater than $m$, it holds that $\delta_m>\delta$. Otherwise, $\delta_m\leq \delta$. In contrast to higher dimensions, the optimal $m$-simplification for any~$m$ of a 1-dimensional curve $P$ has distance $\delta_m\in\{|P(i)-P(j)|/2\mid i,j\in \{1, \dots, n\}\}$ to~$P$, where~the distance is only defined by two vertices of $P$. By Frederickson and Johnson~\cite{FJ84}, we can search over all possible values for $\delta_m$ using $\mathcal{O}(\log n)$ calls of \Cref{Alg:deltaSimpl} and get the following~lemma.

\begin{lemma}
    We can compute an optimal $m$-simplification $P^*$ of $P$ together with a discrete matching that realizes the discrete \F between $P^*$ and $P$ in $\mathcal{O}(n\log n)$~time.
\end{lemma}
\begin{figure}
    \centering
    \includegraphics{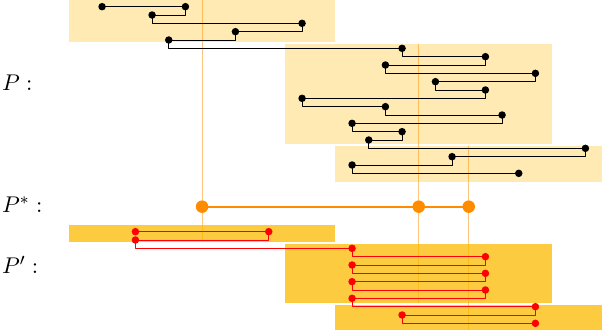}
    \caption{The black curve is $P$ and the middle curve is an optimal $3$-simplification for $P$. The red curve below is the compressed curve $P'$. Here, $k_1=3$, $k_2=7$, $k_3=3$ and $z_1=-1$, $z_2=-1$, $z_3=1$.}
    \label{fig:Preprocessing}
\end{figure}

Let $P^*=\langle p_1, \dots, p_\ell\rangle$ be the optimal $m$-simplification of $P$. Consider the discrete matching~$M$ between $P$ and $P^*$ computed by \Cref{Alg:deltaSimpl}, and let $P(i_j)$ be the last vertex of $P$ that is matched to $p_j$; i.e. \[M=((1,1), (2,1) \dots, (i_1, 1), (i_1+1, 2), ({i_1+2}, 2),\dots, (i_2, 2), ({i_2+1}, 3), \dots, (i_\ell, \ell)).\]
We construct a curve~$P'$ that has discrete \F at most $\delta_m/2$ to $P$ and can be stored in a compressed way using $\mathcal{O}(m)$ space (see \Cref{fig:Preprocessing}). We call~$P'$ a \emph{compressed simplification}.
For every $(x, y)\in M$, define $r_x=p_y-\delta_m/2$ if $P(x)<p_y$ and $r_x=p_y+\delta_m/2$ otherwise. Then, it holds that $|r_x-P(x)|\leq \delta_m/2$. The compressed simplification~$P'$ is $\langle r_1, \dots, r_n\rangle$ after deleting the vertices with identical predecessor, so that~$P'$ contains $r_i$ only if $r_i\neq r_{i-1}$. 
For every $j$, we count the total number $k_j$ that $P'$ contains $p_j\pm\delta_m/2$ and define $z_j\in \{-1, 1\}$ such that the first time $P'$ contains $p_j+z_j \delta_m/2$ is before~$P'$ contains $p_j-z_j \delta_m/2$. Therefore, the values $p_1, \dots, p_\ell$, $k_1, \dots, k_\ell$, $z_1, \dots, z_\ell$, and~$\delta_m$ fully describe $P'$. The next lemma follows directly by construction and $\ell\leq m$.

\begin{lemma}\label{lem:distP'P}
    Given $P$ and its optimal $m$-simplification $P^*$, we can compute in $\mathcal{O}(n)$ time a curve $P'$ such that $\ddF(P, P')\leq \delta_m/2$ and $P'$ can be stored using $\mathcal{O}(m)$ space.
\end{lemma}
\subparagraph*{The Query Algorithm.}
In the following, we give a $2$-approximation query algorithm for the discrete \F between $P$ and any query curve $Q$ of complexity~$m$. We use the precomputed compressed simplification $P'$ and compute $\ddF(P', Q)$ if it is at least~$\frac{3}{2}\delta_m$. 

\begin{lemma}\label{lem:triangleInequPP'Q}
    If $\ddF(P', Q)\leq \frac{3}{2}\delta_m$, then $\ddF(P, Q)\in [\delta_m, 2\delta_m]$.
    Otherwise, define~$\delta$ such that $\frac{3}{2}\delta=\ddF(P', Q)$. Then, $\ddF(P, Q)\in [\delta, 2\delta)$.
\end{lemma}
\begin{proof}
    By definition of $\delta_m$, it holds that $\ddF(P,Q)\geq \delta_m$ for every curve $Q$ of complexity~$m$. Further, by \Cref{lem:distP'P} it holds that $\ddF(P,P')\leq \delta_m/2$.
    Hence, if $\ddF(P', Q)\leq \frac{3}{2}\delta_m$, then
    \[\delta_m\leq \ddF(P,Q)\leq \ddF(P', Q)+\ddF(P,P')\leq \frac{3}{2}\delta_m+\frac{1}{2}\delta_m\leq 2\delta_m.\]
    Otherwise, it holds that $\frac{3}{2}\delta=\ddF(P', Q)>\frac{3}{2}\delta_m$. Hence, $\delta>\delta_m$. Then, 
    \[\delta<\frac{3}{2}\delta-\frac{1}{2}\delta_m\leq \ddF(P', Q)-\ddF(P', P)\leq \ddF(P,Q)\leq \ddF(P', Q)+\ddF(P', P)< 2\delta.\qedhere\]
\end{proof}

Using \Cref{lem:triangleInequPP'Q}, we could simply compute $\ddF(P', Q)$ to get a $2$-approximation algorithm for $\ddF(P,Q)$. However, computing $\ddF(P', Q)$ with known discrete Fréchet distance algorithms takes roughly $\mathcal{O}(m\cdot \sum_{i=1}^{\ell}k_i)$ time and $\sum_{i=1}^{\ell}k_i$ could be $n$. 
In the following, we make use of the special structure of $P'$ to get a faster algorithm that decides whether $\ddF(P', Q)\leq \frac{3}{2}\delta$ for any value $\delta\geq \delta_m$.

We use the concept of the free space matrix $M_{(3/2)\delta}$, but do not compute all reachable entries explicitly.
Define $s_i=\sum_{j=1}^{i-1}k_j+1$ to be the index of the first vertex with value $p_i\pm\delta_m/2$ in $P'$ and $s_{\ell+1}=|P'|+1$.
\Cref{Alg:DecisionP'Q} iterates over the sub-curves $P'[s_i, s_{i+1}-1]$ of $P'$ from~$i=1$ to~$\ell$ and stores in the set~$\FreeOld$ the indices~$j$ of the vertices of $Q$ such that $(s_i-1, j)$ is reachable (see the hatched cells in~\Cref{fig:decision_Algo}). 
Using $\FreeOld$, in iteration~$i$ we successively compute the set $\FreeNew$ of indices such that $(s_{i+1}-1, j)$ is reachable and make use of the following~property.

\begin{observation}\label{obs:distancesToP'}
    For any point $q\in \mathbb{R}$, it holds that exactly one of the following is true:
    \begin{itemize}
        \item $|q-p_i|\leq \delta$ and $|q-P'(s_i+x)|\leq \frac{3}{2}\delta$ for all $x\in \{0, \dots, k_i-1\}$,  or 
        \item $|q-p_i|\in(\delta, 2\delta]$ and $|q-P'(s_i+x)|\leq \frac{3}{2}\delta$ for all even (resp. odd) $x\in \{0, \dots, k_i-1\}$ and $|q-P'(s_i+x)|> \frac{3}{2}\delta$ for all odd (resp. even) $x\in \{0, \dots, k_i-1\}$, or
        \item $|q-p_i|>2\delta$.
    \end{itemize}
\end{observation}

By \Cref{obs:distancesToP'}, it holds that whenever $|Q(j)-p_i|\leq \delta$, there exists an number $\numbA$ such that $(s_i+x, j)$ with $x\in \{0, \dots, k_i-1\}$ is reachable if and only if $x\geq \numbA$ and whenever $|Q(j)-p_i|> 2\delta$, then $(s_i+x, j)$ is not reachable for any integer $x\in \{0, \dots, k_i-1\}$ and we set $\numbA=\infty$. When $|Q(j)-p_i|\in(\delta, 2\delta]$, things get slightly more complicated. However, we can still store a number $\numbA$ such that $(s_i+x, j)$ is reachable for all $x\in \{\numbA, \numbA+2, \dots\}$ with $x\leq k_i-1$. Since there might be more reachable pairs $(s_i+x, j)$, we store a set $\setA$ of size at most $m$ and a number $\updates$ which can be used to reconstruct all other reachable pairs $(s_i+x, j)$ with $x\in \{0, \dots, \numbA\}$. 
Then, let $B\subset\{0, \dots,\min\{k_i-1, \numbA\}\}$ be the set such that for all $x\leq \min\{k_i-1, \numbA\}$ it holds that $(s_i+x, j-1)$ is reachable if and only if $x\in B$. It holds that $B\neq \emptyset$ only if $|Q(j-1)-p_i|\in(\delta, 2\delta]$. In that case, define $z=0$ (resp. $z=1$) if $Q(j-1)$ and $Q(j)$ lie on the same side (resp. different sides) of $p_i$. Then for any $x+z>0$, it holds that that $(s_i+x+z, j-1)$ is reachable if and only if $x\in B$. Using this property, we store a counter~$\updates$ that counts the sum of the values $z$ until we reach again a vertex with $|Q(j')-p_i|\notin(\delta, 2\delta]$. If $x+z=0$, we additionally check whether $(s_i-1, j-1)$ or $(s_i-1, j)$ is reachable by looking at the set~$\FreeOld$. If one of them is reachable, then $(s_i, j-1)=(s_i+(\updates-\updates), j)$ is reachable and we add the current number~$\updates$ to the stored set~$\setA$. In short, we intend to maintain the following invariant. During iteration~$i$ of the outer for-loop and at the end of iteration~$j$ of the inner for-loop, it holds that
$(s_i+x, j)$ is reachable with $x\in \{0, \dots, k_i-1\}$ if and only if
\begin{align*}
    &\textbf{a) } x\in \{u- z\mid z\in \setA\}\cup \{\numbA\} &\text{ or }
    &&\textbf{b) }x>\numbA \text{ and }|Q(j)-P'(s_i+x)|\leq (2/3)\delta.
\end{align*}
See \Cref{fig:decision_Algo} for an example.
Here, during iteration~$i=3$ of the outer for-loop and at the end of iteration~$j=7$ (resp. $j=8$) of the inner for-loop, it holds that $\setA=[0, 2]$, $\updates=2$, and $\numbA=\infty$ (resp. $\setA=[\ ]$, $\updates=0$, and $\numbA=0$).

\begin{algorithm}
    \caption{Decision Algorithm for $\ddF(P',Q)\leq \frac{3}{2}\delta$ \label{Alg:DecisionP'Q}}
    $\FreeNew\gets \{0\}$\;
    \For{$i=1$ to $\complSimp$}{
        $\setA\gets[\ ]$, $\numbA\gets \infty$, $\updates\gets0$, $\FreeOld\gets\FreeNew$, $\FreeNew\gets[\ ]$\;
        \For{$j=1$ to $m$}{
            \If{$|Q(j)-p_i|>2\delta$\label{li:case>2delta}}{
                $\setA\gets[\ ]$, $\numbA\gets \infty$, $\updates\gets 0$\;
            }
            \If{$|Q(j-1)-Q(j)|>2\delta$ and $|Q(j)-p_i|, |Q(j-1)-p_i|\in (\delta, 2\delta]$\label{li:changenumbAif}}
            {
                $\numbA\gets\numbA+1$, $\updates\gets \updates+1$\label{li:changenumbA}\;
            }
            \If{($j\in \FreeOld$ or $j-1\in \FreeOld$) and $|P'(s_i)-Q(j)|\leq \frac{3}{2}\delta$\label{li:InOld}}{
                Append $\updates$ to $\setA$\;
            }
            \If{$|Q(j)-p_i|\leq \delta$}{
                $\numbA\gets\min\{\numbA, u-\min\{\{x\mid x\in \setA\}\}\cup\{\infty\}\}$\label{li:case<delta}\;
                $\setA\gets [\ ]$, $\updates\gets 0$\label{li:case<deltaAemptyset}\;
            }
            Delete all $x\in A$ with $x\leq u-k_i$\;
            \If{$|P'(s_{i+1}-1)-Q(j)|\leq \frac{3}{2}\delta$ and ($a<k_i$ or $\updates-k_i+1\in \setA$)\label{li:addToFreNewif}}{
                $\FreeNew\gets j$\label{li:addToFreNew}\;
            }    
        }
    }
\end{algorithm}

\begin{figure}
    \centering
    \includegraphics[width=0.47\linewidth, page=2]{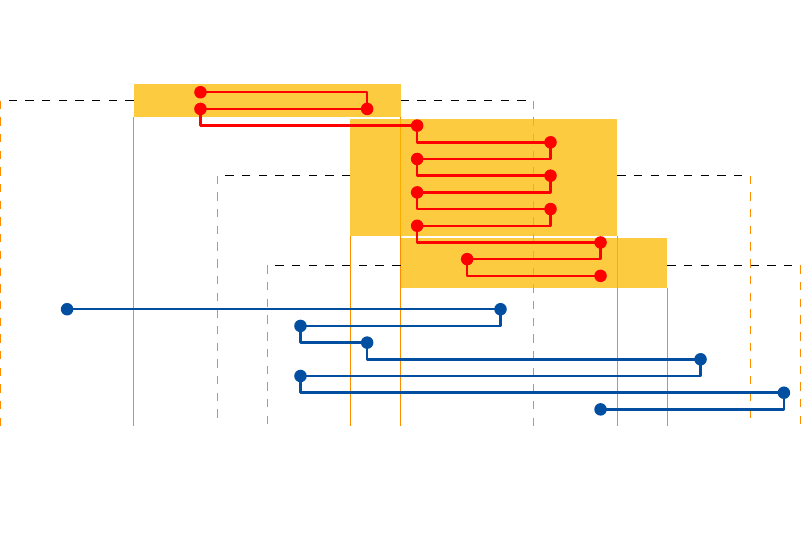}\hspace{0.5cm} 
    \includegraphics[width=0.47\linewidth, page=3]{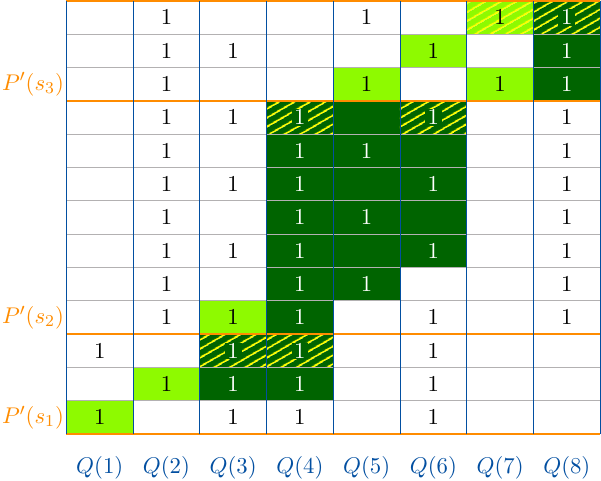}
    \caption{The compressed curve $P'$ and the query curve $Q$ are on the left. On the right is the free space matrix of $P'$ and $Q$, where only $1$-entries are added. The dark green cells visualize the parameter~$a$ and the light green cells the set $\{u-z\mid z\in A\}$. The hatched cells are the cells that have been in $\FreeOld$ and $\FreeNew$.}
    \label{fig:decision_Algo}
\end{figure}

\begin{lemma}\label{lem:PropertyFreeOldNew}
    Assume that at the beginning of iteration $i$ of the outer for-loop after setting $\FreeOld\gets \FreeNew$, it holds $(s_{i}-1, j)$ is reachable if and only if $j\in \FreeOld$.
    Then, at the end of iteration~$i$ it holds that $(s_{i+1}-1, j)$ is reachable if and only if $j\in \FreeNew$.
\end{lemma}
\begin{proof}
    Denote with $\setA_j$, $\numbA_j$, and $\updates_j$ the value of $\setA$, $\numbA$, and $\updates$ at the end of iteration $j$ in the inner for-loop of \Cref{Alg:DecisionP'Q}. Then, we show that the following invariant is true; for every $x\in\{0, \dots, k_i-1\}$, it holds that $(s_i+x, j)$ is reachable if and only if 
    \begin{align*}
    &\textbf{a) } x\in \{\updates_j- z\mid z\in \setA_j\}\cup \{\numbA_j\} &\text{ or }
    &&\textbf{b) }x>\numbA_j \text{ and }|Q(j)-P'(s_i+x)|\leq (2/3)\delta.
\end{align*}
    We prove this invariant via induction on $j$.
    For $j=0$, the invariant clearly holds as $\{u_j-z\mid z\in \setA_j\}\cup \{\numbA_j\}=\{\infty\}$.
    Now assume the invariant holds for $j-1$. 
    Let $x\in\{0, \dots, k_i-1\}$. Then, by definition of the discrete \F $(s_i+x, j)$ is reachable if and only if $|Q(j)-P(s_i+x)|\leq \frac{2}{3}\delta$ and $(s_i+x-1, j)$, $(s_i+x, j-1)$, or $(s_i+x-1, j-1)$ is reachable. 
    We consider three different cases. 
    \begin{enumerate}
        \item If $|Q(j)-p_i|>2\delta$, then $(s_i+x, j)$ is not reachable, since $|Q(j)-P(s_i+x)|> \frac{2}{3}\delta$. By line~\ref{li:case>2delta} of \Cref{Alg:DecisionP'Q}, the invariant holds.
        \item If $|Q(j)-p_i|\leq \delta$, then $|Q(j)-P(s_i+x)|\leq\frac{2}{3}\delta$ for any $x\in \{0, \dots, k_i-1\}$. Further, observe that if $(s_i+y, j)$ is reachable for $y\in \{0, \dots, k_i-1\}$, then $(s_i+x, j)$ is also reachable for every $y\leq x \leq  k_i-1$. Define $y\in \{0, \dots, k_i-1\}$ to be the minimum such that $(s_i+y, j)$ is reachable. 
        Hence, by line~\ref{li:case<deltaAemptyset} it holds that the invariant holds if $a_j=y$. 
        Unless $y=0$ and $(s_i-1, j)$ or $(s_i-1, j-1)$ is reachable, it must hold that $(s_i+y, j-1)$ is reachable as otherwise $y$ would not be the minimum. Further again since $y$ is the minimum and by the invariant for $j-1$, it holds that $y=\min\{u_{j-1}-z\mid z\in \setA_{j-1}\}\cup \{\numbA_{j-1}\}$. 
        Now, consider the special case that $y=0$ and $(s_i-1, j)$, or $(s_i-1, j-1)$ is reachable. By assumption, it holds that $(s_i-1, j)$ (resp. $(s_i-1, j-1)$) is reachable if and only if $j\in \FreeOld$ (resp. $j-1\in \FreeOld$).
        Hence, $a_j$ is set to $y$ by line \ref{li:InOld}--\ref{li:case<delta} and the invariant is true. 
        \item It remains the case that $\delta<|Q(j)-p_i| \leq 2\delta$.
        \begin{enumerate}
            \item First, consider the case that $x\geq \numbA_{j-1}$. Then since $(s_i+\numbA_{j-1}, j-1)$ is reachable by the induction hypothesis, it holds that $|Q(j-1)-p_i|\leq 2\delta$. Hence, $|Q(j-1)-P(s_i+x)|\leq \frac{2}{3}\delta$ or $|Q(j-1)-P(s_i+x-1)|\leq \frac{2}{3}\delta$. Further again by the induction hypothesis it holds that $(s_i+x, j-1)$ or $(s_i+x-1, j-1)$ is reachable. Hence, by line~\ref{li:changenumbAif}--\ref{li:changenumbA} the invariant holds for $x\geq \numbA_{j-1}$.
            \item Otherwise, it holds that $x<\numbA_{j-1}$. Consider the special case that $x=0$ and $(s_i+x, j)$ is reachable because $|Q(j)-P(s_i+x)|\leq \frac{2}{3}\delta$ and $(s_i-1, j)$ or $(s_i-1, j-1)$ are reachable. Then by \ref{li:InOld} and the assumption on $\FreeOld$, it holds that $u_{j}\in A_j$. 
            Otherwise it holds that $(s_i+x, j)$ is reachable if and only if $(s_i+x, j-1)$ or $(s_i+x-1, j-1)$ is reachable and $|Q(j)-p_i|\leq \frac{3}{2}\delta$. This is true because it holds that either ${|Q(j)-P(s_i+x)|\leq \frac{2}{3}\delta}$ or $|Q(j)-P(s_i+x-1)|\leq \frac{2}{3}\delta$ since $\delta<|Q(j)-p_i|\leq 2\delta$.
            Hence, $x-1\in \setA_{j-1}$ or $x\in\setA_{j-1}$.
            Therefore, $\setA_{j-1}\neq[\ ]$ and by the algorithm it holds that $\delta<|Q(j-1)-p_i|\leq 2\delta$.
            If $\min\{Q(j-1), Q(j)\}>p_i$ or $\max\{Q(j-1),Q(j)\}<p_i$, then $(s_i+x, j)$ is reachable if and only if $(s_i+x, j-1)$ is reachable by \Cref{obs:distancesToP'}. Otherwise $(s_i+x, j)$ is reachable if and only if $(s_i+x-1, j-1)$ is reachable. This is ensured by line~\ref{li:changenumbAif}--\ref{li:changenumbA}.
        \end{enumerate}
    \end{enumerate}
    This proves that the invariant is true for every $j$. Note that $(s_{i+1}-1, j)=(s_{i}+k_i-1, j)$. By the invariant it holds that $(s_{i}+k_i-1, j)$ is reachable if and only if $|P'(s_{i+1}-1)-Q(j)|\leq \frac{3}{2}\delta$ and $a_j<k_i$ or $u_j-k_i+1\in A_j$. Hence the lemma follows by line~\ref{li:addToFreNewif}--\ref{li:addToFreNew}.
\end{proof}

\begin{lemma}\label{lem:1Ddecision}
    For any value $\delta\geq \delta_m$, we can decide whether $\ddF(P', Q)\leq \frac{3}{2}\delta$ in $\mathcal{O}(m\ell)$~time. 
\end{lemma}
\begin{proof}
    We run \Cref{Alg:DecisionP'Q} and output $\ddF(P', Q)\leq \frac{3}{2}\delta$ if and only if $m\in \FreeNew$ at the end of \Cref{Alg:DecisionP'Q}. The running time of \Cref{Alg:DecisionP'Q} is in $\mathcal{O}(m\ell)$, since the array $\setA$ is sorted and in every iteration of the inner for loop at most one point gets appended to $\setA$.
    By induction on the index $i$ of $s_i$ and \Cref{lem:PropertyFreeOldNew}, it holds that $(|P'|, m)=(s_{\ell+1}-1, m)$ is reachable if and only if $\ddF(P', Q)\leq \frac{3}{2}\delta$.
\end{proof}

It remains to optimize the~$\delta$ value.
First, we decide whether $\ddF(P', Q)\leq \frac{3}{2}\delta_m$ or not using \Cref{Alg:DecisionP'Q}. If $\ddF(P', Q)\leq \frac{3}{2}\delta_m$, then $\ddF(P,Q)\in [\delta_m, 2\delta_m]$ by \Cref{lem:triangleInequPP'Q} and $\delta_m$ is a $2$-approximation for $\ddF(P,Q)$. 
Otherwise $\ddF(P', Q)>\frac{3}{2}\delta_m$. In this case, we use the observation that the discrete \F between $P'$ and $Q$ gets realizes by a point to point distance between a vertex of $P'$ and a vertex of $Q$. Hence, 
\[\ddF(P', Q)\in\{|(p_i\pm\delta/2)-Q(j)|\mid i\in \{1, \dots, \ell\}, j\in \{1, \dots, m\}\}.\]
We sort these values in $\mathcal{O}(m^2\log m)$ time and search over the ones that are greater~$\frac{3}{2}\delta_m$. Using \Cref{Alg:DecisionP'Q}, we compute in $\mathcal{O}(m^2\log m)$ time $\frac{3}{2}\delta=\ddF(P', Q)$. 
Then, it holds that $\ddF(P,Q)\in [\delta, 2\delta]$ by \Cref{lem:triangleInequPP'Q}.
Therefore, we get the following theorem:

\begin{theorem}\label{thm:1DUpperDS}
    For a given value $m$, there exists a data structure that preprocesses a 1-dimensional curve $P$ of complexity $n$ in $\mathcal{O}(n\log n)$ time using $\mathcal{O}(m)$ space such that for any 1-dimensional query curve $Q$ of complexity at most $m$, we can approximate the discrete \F between $P$ and $Q$ within a factor of $2$ in $\mathcal{O}(m^2\log m)$ time.
\end{theorem}

The next corollary is a direct consequence of the theorem above.
\begin{corollary}\label{cor:1DUpper}
    Given two 1-dimensional curves $P$ and $Q$ of complexity $n$ and $m$ with $n\geq m$. There exists a $2$-approximation algorithm to compute the discrete \F between $P$ and $Q$ in $\mathcal{O}(n\log n+m^2\log m)$ time. 
\end{corollary}
By \Cref{thm:1D}, it follows that there does not exist an algorithm using $\mathcal{O}(n\log n+m^2\log m)$ time with a better approximation factor. Further, the lower bound in~\cite{BOS18} shows the almost optimality of the running time.
\section{\texorpdfstring{\boldmath $(3+\eps)$-approximation Algorithm for the Fréchet distance}{3+eps-approximation Algorithm for the Fréchet distance}}
In this section, we give a $(3+\eps)$-approximation algorithm to compute the continuous/discrete \F between two curves $P$ and $Q$ of complexity $n$ resp. $m$ with $m\leq n$. We describe this section for the continuous \F, but everything works analogously for the discrete \F as well.
We begin with the decision variant for~$\delta$. In the first step, using the curve~$Q$, we compute a curve~$P'$ of complexity at most $2m$ such that $\dF(P, P')\leq \delta$ or decide that $\dF(P, Q)>\delta$ in linear time.  

\Cref{Alg:simplwithQ} computes such a curve~$P'$. The curve~$P'$ consists of concatenated sub-edges of~$Q$, and their order along~$P'$ is the same as along~$Q$. The algorithm proceeds iteratively until the end of $P$ or $Q$ is reached. At the beginning of the $i$-th iteration, we have constructed a prefix $\langle Q(s_1), Q(t_1), Q(s_2), \dots Q(s_i), Q(t_i)\rangle$ of $P'$ whose \F to the prefix $P[1, a_i]$ of $P$ at most~$\delta$. We then search for the earliest point on $Q$ that lies within distance~$\delta$ of $P(a_i)$ that comes later on $Q$ than~$Q(t_i)$. If no such point exists, we show in \Cref{lem:simplWithQ} that $\dF(P,Q)>\delta$. Otherwise, we find the longest subcurve~$P[a_i, a_{i+1}]$ of~$P$, starting at~$P(a_i)$, within distance~$\delta$ to a part of the edge containing $Q(s_i)$. If $a_{i+1}=n$, then we found a curve~$P'$ of complexity at most $2m$ with distance at most~$\delta$ to $P$. 
\begin{algorithm}
    \caption{Simplification of $P$ depending on $Q$ \label{Alg:simplwithQ}}
    $s_0\gets0$, $a_1\gets 1$, $i\gets1$\;
    \While{True}{
        $s_{i}\gets \min\{s\geq \lfloor s_{i-1}+1\rfloor\mid \|P(a_{i})-Q(s)\|\leq \delta\}\cup \{\infty\}$\label{li:increaseS_i}\;
        \If{$s_{i}=\infty$}{
            Return $\dF(P, Q)>\delta$\;
        }
        Find largest $a_{i+1}$ such that $\exists\ t_i\leq \lfloor s_i+1\rfloor$ with $\dF(P[a_i, a_{i+1}], \overline{Q(s_i) Q(t_i)})\leq\delta$\label{li:findLargesA_i+1}\;
        \If{$a_{i+1}=n$}{
            Return $P'=\langle Q(s_1), Q(t_1), Q(s_2), Q(t_2), \dots, Q(s_i), Q(t_i)\rangle$
        }
        $i\gets i+1$\;
    }
\end{algorithm}

\begin{lemma}\label{lem:simplWithQ}
    \Cref{Alg:simplwithQ} computes a curve $P'$ of complexity at most $2m$ with $\dF(P, P')\leq \delta$ or returns $\dF(P, Q)>\delta$ in $\Oh(n)$ time.
\end{lemma}
\begin{proof}
    It holds that $s_i\geq i$ by line~\ref{li:increaseS_i} of \Cref{Alg:simplwithQ}. Further, it holds that $a_i\leq a_{i+1}$ for every~$i$ and line~\ref{li:findLargesA_i+1} has running time $\Oh(a_{i+1}-a_i)$ as $\overline{Q(s_i)Q(t_i)}\subset\overline{Q(s_i)Q(\lfloor s_i+1\rfloor)}$. Hence, the running time follows. 
    If a curve~$P'$ is returned, it holds that $\dF(P, P')\leq \delta$ by line~\ref{li:findLargesA_i+1} and \Cref{o:concatanation}. The complexity of~$P'$ is at most $2m$ since there are at most~$m$ iterations of the while loop. 
    It remains to show correctness for the case when $\dF(P, Q)>\delta$ is returned. Assume that it holds $\dF(P, Q)\leq\delta$. Let $M$ be the matching between~$P$ and $Q$ that realizes the \F. We consider two cases (see \Cref{fig:simplification_depending_on_Q}). In the first case, there exist $b, b'$ and an index $i$ such that $(b, s_i), (b', t_i)\in M$ with $b\leq a_i$ and $b'>a_{i+1}$. This contradicts line~\ref{li:findLargesA_i+1} of the algorithm. Otherwise, there must exist an index $i$ and a parameter~$c$ such that $(a_i, c)\in M$ with $t_{i-1}\leq c<s_i$. This contradicts line~\ref{li:increaseS_i}. Hence, it follows that $\dF(P, Q)>\delta$.
\end{proof}

\begin{figure}
    \centering
    \includegraphics[page=1]{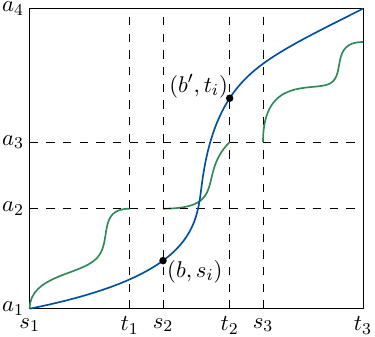}\hspace{1cm}
    \includegraphics[page=2]{Figures/Simplification_depending_on_Q.pdf}
    \caption{The green curve is a matching between the curves $P[a_i, a_{i+1}]$ and $Q[s_i, t_i]$ computed in \Cref{Alg:simplwithQ}. The blue curves visualize the matching used in the proof of \Cref{lem:simplWithQ}.}
    \label{fig:simplification_depending_on_Q}
\end{figure}

\Cref{lem:simplWithQ} can be easily adapted such that it holds for the discrete \F. 
Given $P'$, we decide whether $\dF(P',Q)\leq 2\delta$ in $\Oh(m^2)$ time \cite{AG95} or slightly faster~\cite{CH25}. If $\dF(P',Q)\leq 2\delta$, then $\dF(P,Q)\leq \dF(P',Q)+\dF(P, P')\leq 3\delta$. Otherwise, it holds that $\dF(P,Q)\geq \dF(P',Q)-\dF(P, P')>\delta$. Hence, we obtain the following lemma.

\begin{lemma}\label{lem:decision_3-approx}
    Given two curves $P$ and $Q$ of complexity $n$ and $m$ in any metric space~$M$. Assume that we can compute distances in $M$ in constant time. 
    We can $3$-approximate the decision variant of the continuous/discrete \F in $\Oh(n+m^2)$ time.
\end{lemma}

For curves in the Euclidean space, Colombe and Fox (Theorem 11 of \cite{CF21}) show how to transform a decision algorithm into an approximate optimization algorithm increasing the running time by only $\Oh(\log(n/\eps))$ factor. Since all $L_p$-norms on $\mathbb{R}^d$ are equivalent and \Cref{lem:decision_3-approx} holds for all $L_p$-norms, we can adopt their approach such that we can use it for all $L_p$-norms.

\begin{theorem}\label{thm:dDUpper}
    There exist a $(3+\eps)$-approximation algorithm for the continuous/discrete \F using $\Oh((n+m^2) \log(n/\eps))$ time.
\end{theorem}
\bibliography{literature.bib}

@inproceedings{BD24,
  author       = {Lotte Blank and
                  Anne Driemel},
  editor       = {Timothy M. Chan and
                  Johannes Fischer and
                  John Iacono and
                  Grzegorz Herman},
  title        = {A Faster Algorithm for the {F}r{\'{e}}chet Distance in 1D for the Imbalanced Case},
  booktitle    = {32nd Annual European Symposium on Algorithms, {ESA} 2024, September
                  2-4, 2024, Royal Holloway, London, United Kingdom},
  series       = {LIPIcs},
  volume       = {308},
  pages        = {28:1--28:15},
  publisher    = {Schloss Dagstuhl - Leibniz-Zentrum f{\"{u}}r Informatik},
  year         = {2024},
  url          = {https://doi.org/10.4230/LIPIcs.ESA.2024.28},
  doi          = {10.4230/LIPICS.ESA.2024.28},
  bibsource    = {dblp computer science bibliography, https://dblp.org}
}

@inproceedings{BC19,
  author       = {Karl Bringmann and
                  Bhaskar Ray Chaudhury},
  editor       = {Gill Barequet and
                  Yusu Wang},
  title        = {Polyline Simplification has Cubic Complexity},
  booktitle    = {35th International Symposium on Computational Geometry, SoCG 2019,
                  June 18-21, 2019, Portland, Oregon, {USA}},
  series       = {LIPIcs},
  volume       = {129},
  pages        = {18:1--18:16},
  publisher    = {Schloss Dagstuhl - Leibniz-Zentrum f{\"{u}}r Informatik},
  year         = {2019},
  url          = {https://doi.org/10.4230/LIPIcs.SoCG.2019.18},
  doi          = {10.4230/LIPICS.SOCG.2019.18},
  bibsource    = {dblp computer science bibliography, https://dblp.org}
}

@inproceedings{BJWYZ08,
  author       = {Sergey Bereg and
                  Minghui Jiang and
                  Wencheng Wang and
                  Boting Yang and
                  Binhai Zhu},
  editor       = {Eduardo Sany Laber and
                  Claudson F. Bornstein and
                  Loana Tito Nogueira and
                  Lu{\'{e}}rbio Faria},
  title        = {Simplifying 3D Polygonal Chains Under the Discrete {F}r{\'{e}}chet
                  Distance},
  booktitle    = {{LATIN} 2008: Theoretical Informatics, 8th Latin American Symposium,
                  B{\'{u}}zios, Brazil, April 7-11, 2008, Proceedings},
  series       = {Lecture Notes in Computer Science},
  volume       = {4957},
  pages        = {630--641},
  publisher    = {Springer},
  year         = {2008},
  url          = {https://doi.org/10.1007/978-3-540-78773-0\_54},
  doi          = {10.1007/978-3-540-78773-0\_54},
  bibsource    = {dblp computer science bibliography, https://dblp.org}
}

@inproceedings{KLW18,
  author       = {Marc J. van Kreveld and
                  Maarten L{\"{o}}ffler and
                  Lionov Wiratma},
  editor       = {Bettina Speckmann and
                  Csaba D. T{\'{o}}th},
  title        = {On Optimal Polyline Simplification Using the {H}ausdorff and {F}r{\'{e}}chet
                  Distance},
  booktitle    = {34th International Symposium on Computational Geometry, SoCG 2018,
                  June 11-14, 2018, Budapest, Hungary},
  series       = {LIPIcs},
  volume       = {99},
  pages        = {56:1--56:14},
  publisher    = {Schloss Dagstuhl - Leibniz-Zentrum f{\"{u}}r Informatik},
  year         = {2018},
  url          = {https://doi.org/10.4230/LIPIcs.SoCG.2018.56},
  doi          = {10.4230/LIPICS.SOCG.2018.56},
  bibsource    = {dblp computer science bibliography, https://dblp.org}
}

@article{AHMW05,
  author       = {Pankaj K. Agarwal and
                  Sariel Har{-}Peled and
                  Nabil H. Mustafa and
                  Yusu Wang},
  title        = {Near-Linear Time Approximation Algorithms for Curve Simplification},
  journal      = {Algorithmica},
  volume       = {42},
  number       = {3-4},
  pages        = {203--219},
  year         = {2005},
  url          = {https://doi.org/10.1007/s00453-005-1165-y},
  doi          = {10.1007/S00453-005-1165-Y},
  bibsource    = {dblp computer science bibliography, https://dblp.org}
}

@article{DPS19,
  author       = {Anne Driemel and
                  Ioannis Psarros and
                  Melanie Schmidt},
  title        = {Sublinear data structures for short {F}r{\'{e}}chet queries},
  journal      = {CoRR},
  volume       = {abs/1907.04420},
  year         = {2019},
  url          = {http://arxiv.org/abs/1907.04420},
  eprinttype    = {arXiv},
  eprint       = {1907.04420},
  bibsource    = {dblp computer science bibliography, https://dblp.org}
}

@article{Wil05,
  author       = {Ryan Williams},
  title        = {A new algorithm for optimal 2-constraint satisfaction and its implications},
  journal      = {Theor. Comput. Sci.},
  volume       = {348},
  number       = {2-3},
  pages        = {357--365},
  year         = {2005},
  url          = {https://doi.org/10.1016/j.tcs.2005.09.023},
  doi          = {10.1016/J.TCS.2005.09.023},
  bibsource    = {dblp computer science bibliography, https://dblp.org}
}

@inproceedings{BK18,
  author       = {Karl Bringmann and
                  Marvin K{\"{u}}nnemann},
  editor       = {Artur Czumaj},
  title        = {Multivariate Fine-Grained Complexity of Longest Common Subsequence},
  booktitle    = {Proceedings of the Twenty-Ninth Annual {ACM-SIAM} Symposium on Discrete
                  Algorithms, {SODA} 2018, New Orleans, LA, USA, January 7-10, 2018},
  pages        = {1216--1235},
  publisher    = {{SIAM}},
  year         = {2018},
  url          = {https://doi.org/10.1137/1.9781611975031.79},
  doi          = {10.1137/1.9781611975031.79},
  bibsource    = {dblp computer science bibliography, https://dblp.org}
}

@inproceedings{R18,
  author       = {Aviad Rubinstein},
  editor       = {Ilias Diakonikolas and
                  David Kempe and
                  Monika Henzinger},
  title        = {Hardness of approximate nearest neighbor search},
  booktitle    = {Proceedings of the 50th Annual {ACM} {SIGACT} Symposium on Theory
                  of Computing, {STOC} 2018, Los Angeles, CA, USA, June 25-29, 2018},
  pages        = {1260--1268},
  publisher    = {{ACM}},
  year         = {2018},
  url          = {https://doi.org/10.1145/3188745.3188916},
  doi          = {10.1145/3188745.3188916},
  bibsource    = {dblp computer science bibliography, https://dblp.org}
}

@inproceedings{CHJ25,
  author       = {Siu{-}Wing Cheng and
                  Haoqiang Huang and
                  Le Jiang},
  editor       = {Oswin Aichholzer and
                  Haitao Wang},
  title        = {Simplification of Trajectory Streams},
  booktitle    = {41st International Symposium on Computational Geometry, SoCG 2025,
                  June 23-27, 2025, Kanazawa, Japan},
  series       = {LIPIcs},
  volume       = {332},
  pages        = {34:1--34:14},
  publisher    = {Schloss Dagstuhl - Leibniz-Zentrum f{\"{u}}r Informatik},
  year         = {2025},
  url          = {https://doi.org/10.4230/LIPIcs.SoCG.2025.34},
  doi          = {10.4230/LIPICS.SOCG.2025.34},
  bibsource    = {dblp computer science bibliography, https://dblp.org}
}

@inproceedings{CH23,
  author       = {Siu{-}Wing Cheng and
                  Haoqiang Huang},
  editor       = {Nikhil Bansal and
                  Viswanath Nagarajan},
  title        = {Curve Simplification and Clustering under {F}r{\'{e}}chet Distance},
  booktitle    = {Proceedings of the 2023 {ACM-SIAM} Symposium on Discrete Algorithms,
                  {SODA} 2023, Florence, Italy, January 22-25, 2023},
  pages        = {1414--1432},
  publisher    = {{SIAM}},
  year         = {2023},
  url          = {https://doi.org/10.1137/1.9781611977554.ch51},
  doi          = {10.1137/1.9781611977554.CH51},
  bibsource    = {dblp computer science bibliography, https://dblp.org}
}

@article{EM94,
author = {Eiter, Thomas and Mannila, Heikki},
year = {1994},
month = {05},
pages = {},
title = {Computing Discrete {F}réchet Distance}
}

@inproceedings{CHZ25,
  author       = {Siu{-}Wing Cheng and
                  Haoqiang Huang and
                  Shuo Zhang},
  editor       = {Michal Kouck{\'{y}} and
                  Nikhil Bansal},
  title        = {Constant Approximation of {F}r{\'{e}}chet Distance in Strongly
                  Subquadratic Time},
  booktitle    = {Proceedings of the 57th Annual {ACM} Symposium on Theory of Computing,
                  {STOC} 2025, Prague, Czechia, June 23-27, 2025},
  pages        = {2329--2340},
  publisher    = {{ACM}},
  year         = {2025},
  url          = {https://doi.org/10.1145/3717823.3718157},
  doi          = {10.1145/3717823.3718157},
  bibsource    = {dblp computer science bibliography, https://dblp.org}
}

@inproceedings{CH25,
  author       = {Siu{-}Wing Cheng and
                  Haoqiang Huang},
  editor       = {Yossi Azar and
                  Debmalya Panigrahi},
  title        = {{F}r{\'{e}}chet Distance in Subquadratic Time},
  booktitle    = {Proceedings of the 2025 Annual {ACM-SIAM} Symposium on Discrete Algorithms,
                  {SODA} 2025, New Orleans, LA, USA, January 12-15, 2025},
  pages        = {5100--5113},
  publisher    = {{SIAM}},
  year         = {2025},
  url          = {https://doi.org/10.1137/1.9781611978322.173},
  doi          = {10.1137/1.9781611978322.173},
  bibsource    = {dblp computer science bibliography, https://dblp.org}
}

@article{GSW24,
  author       = {Joachim Gudmundsson and
                  Martin P. Seybold and
                  Sampson Wong},
  title        = {Map Matching Queries on Realistic Input Graphs Under the {F}r{\'{e}}chet
                  Distance},
  journal      = {{ACM} Trans. Algorithms},
  volume       = {20},
  number       = {2},
  pages        = {14},
  year         = {2024},
  url          = {https://doi.org/10.1145/3643683},
  doi          = {10.1145/3643683},
  bibsource    = {dblp computer science bibliography, https://dblp.org}
}

@inproceedings{BDNP21,
  author       = {Karl Bringmann and
                  Anne Driemel and
                  Andr{\'{e}} Nusser and
                  Ioannis Psarros},
  editor       = {Joseph (Seffi) Naor and
                  Niv Buchbinder},
  title        = {Tight Bounds for Approximate Near Neighbor Searching for Time Series
                  under the {F}r{\'{e}}chet Distance},
  booktitle    = {Proceedings of the 2022 {ACM-SIAM} Symposium on Discrete Algorithms,
                  {SODA} 2022, Virtual Conference / Alexandria, VA, USA, January 9 -
                  12, 2022},
  pages        = {517--550},
  publisher    = {{SIAM}},
  year         = {2022},
  doi          = {10.1137/1.9781611977073.25},
  bibsource    = {dblp computer science bibliography, https://dblp.org}
}

@inproceedings{BOS18,
  author       = {Kevin Buchin and
                  Tim Ophelders and
                  Bettina Speckmann},
  editor       = {Timothy M. Chan},
  title        = {{SETH} Says: Weak {F}r{\'{e}}chet Distance is Faster, but only
                  if it is Continuous and in One Dimension},
  booktitle    = {Proceedings of the Thirtieth Annual {ACM-SIAM} Symposium on Discrete
                  Algorithms, {SODA} 2019, San Diego, California, USA, January 6-9,
                  2019},
  pages        = {2887--2901},
  publisher    = {{SIAM}},
  year         = {2019},
  doi          = {10.1137/1.9781611975482.179},
  bibsource    = {dblp computer science bibliography, https://dblp.org}
}

@article{AG95,
  author       = {Helmut Alt and
                  Michael Godau},
  title        = {Computing the {F}r{\'{e}}chet distance between two polygonal curves},
  journal      = {Int. J. Comput. Geom. Appl.},
  volume       = {5},
  pages        = {75--91},
  year         = {1995},
  doi          = {10.1142/S0218195995000064},
  bibsource    = {dblp computer science bibliography, https://dblp.org}
}

@article{FJ84,
  author       = {Greg N. Frederickson and
                  Donald B. Johnson},
  title        = {Generalized Selection and Ranking: Sorted Matrices},
  journal      = {{SIAM} J. Comput.},
  volume       = {13},
  number       = {1},
  pages        = {14--30},
  year         = {1984},
  url          = {https://doi.org/10.1137/0213002},
  doi          = {10.1137/0213002},
  bibsource    = {dblp computer science bibliography, https://dblp.org}
}

@article{FF23,
  author       = {Arnold Filtser and
                  Omrit Filtser},
  title        = {Static and Streaming Data Structures for {F}r{\'{e}}chet Distance
                  Queries},
  journal      = {{ACM} Trans. Algorithms},
  volume       = {19},
  number       = {4},
  pages        = {39:1--39:36},
  year         = {2023},
  url          = {https://doi.org/10.1145/3610227},
  doi          = {10.1145/3610227},
  bibsource    = {dblp computer science bibliography, https://dblp.org}
}

@article{DH13,
  author       = {Anne Driemel and
                  Sariel Har{-}Peled},
  title        = {Jaywalking Your Dog: Computing the {F}r{\'{e}}chet Distance with Shortcuts},
  journal      = {{SIAM} J. Comput.},
  volume       = {42},
  number       = {5},
  pages        = {1830--1866},
  year         = {2013},
  url          = {https://doi.org/10.1137/120865112},
  doi          = {10.1137/120865112},
  bibsource    = {dblp computer science bibliography, https://dblp.org}
}

@inproceedings{CH24,
  author       = {Siu{-}Wing Cheng and
                  Haoqiang Huang},
  editor       = {David P. Woodruff},
  title        = {Solving {F}r{\'{e}}chet Distance Problems by Algebraic Geometric
                  Methods},
  booktitle    = {Proceedings of the 2024 {ACM-SIAM} Symposium on Discrete Algorithms,
                  {SODA} 2024, Alexandria, VA, USA, January 7-10, 2024},
  pages        = {4502--4513},
  publisher    = {{SIAM}},
  year         = {2024},
  url          = {https://doi.org/10.1137/1.9781611977912.158},
  doi          = {10.1137/1.9781611977912.158},
  bibsource    = {dblp computer science bibliography, https://dblp.org}
}

@inproceedings{CF21,
  author       = {Connor Colombe and
                  Kyle Fox},
  editor       = {Kevin Buchin and
                  {\'{E}}ric Colin de Verdi{\`{e}}re},
  title        = {Approximating the (Continuous) {F}r{\'{e}}chet Distance},
  booktitle    = {37th International Symposium on Computational Geometry, SoCG 2021,
                  June 7-11, 2021, Buffalo, NY, {USA} (Virtual Conference)},
  series       = {LIPIcs},
  volume       = {189},
  pages        = {26:1--26:14},
  publisher    = {Schloss Dagstuhl - Leibniz-Zentrum f{\"{u}}r Informatik},
  year         = {2021},
  url          = {https://doi.org/10.4230/LIPIcs.SoCG.2021.26},
  doi          = {10.4230/LIPICS.SOCG.2021.26},
  bibsource    = {dblp computer science bibliography, https://dblp.org}
}

@article{HO24,
  author       = {Thijs van der Horst and
                  Tim Ophelders},
  title        = {Faster {F}r{\'{e}}chet Distance Approximation through Truncated
                  Smoothing},
  journal      = {CoRR},
  volume       = {abs/2401.14815},
  year         = {2024},
  url          = {https://doi.org/10.48550/arXiv.2401.14815},
  doi          = {10.48550/ARXIV.2401.14815},
  eprinttype    = {arXiv},
  eprint       = {2401.14815},
  bibsource    = {dblp computer science bibliography, https://dblp.org}
}

@inproceedings{HKOS23,
  author       = {Thijs van der Horst and
                  Marc J. van Kreveld and
                  Tim Ophelders and
                  Bettina Speckmann},
  editor       = {Nikhil Bansal and
                  Viswanath Nagarajan},
  title        = {A Subquadratic \emph{n}\({}^{\mbox{{\(\epsilon\)}}}\)-approximation for the Continuous {F}r{\'{e}}chet Distance},
  booktitle    = {Proceedings of the 2023 {ACM-SIAM} Symposium on Discrete Algorithms,
                  {SODA} 2023, Florence, Italy, January 22-25, 2023},
  pages        = {1759--1776},
  publisher    = {{SIAM}},
  year         = {2023},
  url          = {https://doi.org/10.1137/1.9781611977554.ch67},
  doi          = {10.1137/1.9781611977554.CH67},
  bibsource    = {dblp computer science bibliography, https://dblp.org}
}

@article{AAK2014,
author = {Agarwal, Pankaj K. and Avraham, Rinat Ben and Kaplan, Haim and Sharir, Micha},
title = {Computing the Discrete {F}réchet Distance in Subquadratic Time},
journal = {SIAM Journal on Computing},
volume = {43},
number = {2},
pages = {429-449},
year = {2014},
doi = {10.1137/130920526}
}

@article{BBMM2017,
author = {Buchin, Kevin and Buchin, Maike and Meulemans, Wouter and Mulzer, Wolfgang},
title = {Four Soviets Walk the Dog: Improved Bounds for Computing the {F}r\'{e}chet Distance},
year = {2017},
issue_date = {July      2017},
publisher = {Springer-Verlag},
address = {Berlin, Heidelberg},
volume = {58},
number = {1},
issn = {0179-5376},
url = {https://doi.org/10.1007/s00454-017-9878-7},
doi = {10.1007/s00454-017-9878-7},
journal = {Discrete Comput. Geom.},
month = {jul},
pages = {180–216},
numpages = {37}
}

@inproceedings{bringmann2014walking,
  title={Why walking the dog takes time: {F}r\'{e}chet distance has no strongly subquadratic algorithms unless SETH fails},
  author={Bringmann, Karl},
  booktitle={2014 IEEE 55th Annual Symposium on Foundations of Computer Science},
  pages={661--670},
  year={2014},
  organization={IEEE}
}

@article{bringmann2016approximability,
  title={Approximability of the discrete {F}r{\'e}chet distance},
  author={Bringmann, Karl and Mulzer, Wolfgang},
  journal={Journal of Computational Geometry},
  volume={7},
  number={2},
  pages={46--76},
  year={2016}
}
\end{document}